\documentclass[letterpaper, 10 pt, conference]{ieeeconf}
\IEEEoverridecommandlockouts
\usepackage{graphics} 
\usepackage{epsfig} 
\usepackage{mathrsfs}
\usepackage{times} 
\usepackage{amsmath} 
\usepackage{amssymb}  
\usepackage{graphicx}
\usepackage{tabularx,array,ragged2e,booktabs}
\usepackage{color}
\usepackage{lipsum}
\usepackage{subcaption}
\usepackage{booktabs}
\usepackage{xcolor}
\usepackage{pifont}

\usepackage{tikz}
\usepackage{pgfplots}
\pgfplotsset{compat=1.18}

\usepackage{booktabs}
\usepackage{siunitx}
\newtheorem{assumption}{Assumption}
\newtheorem{corollary}{Corollary}
\newtheorem{proposition}{Proposition}
\newtheorem{lemma}{Lemma}
\newtheorem{theorem}{Theorem}
\usepackage[style=ieee,dashed=false]{biblatex}

\title{\LARGE \bf
Finite Memory Belief Approximation for Optimal Control in \\ Partially Observable Markov Decision Processes
}

\author{Mintae Kim
\thanks{The author is with \textit{Hybrid Robotics Lab}, University of California, Berkeley, CA 94720, United States.}
\thanks{E-mail: \texttt{mintae.kim@berkeley.edu}}
}

\begin{document}

\maketitle
\thispagestyle{empty}
\pagestyle{empty}


\begin{abstract}

We study finite memory belief approximation for partially observable (PO) stochastic optimal control (SOC) problems.
While belief states are sufficient for SOC in partially observable Markov decision processes (POMDPs), they are generally infinite-dimensional and impractical.
We interpret truncated input-output (IO) histories as inducing a belief approximation and develop a metric-based theory that directly relates information loss to control performance.
Using the Wasserstein metric, we derive policy-conditional performance bounds that quantify value degradation induced by finite memory along typical closed-loop trajectories.
Our analysis proceeds via a fixed-policy comparison: we evaluate two cost functionals under the same closed-loop execution and isolate the effect of replacing the true belief by its finite memory approximation inside the belief-level cost.
For linear quadratic Gaussian (LQG) systems, we provide closed-form belief mismatch evaluation and empirically validate the predicted mechanism, demonstrating that belief mismatch decays approximately exponentially with memory length and that the induced performance mismatch scales accordingly.
Together, these results provide a metric-aware characterization of what finite memory belief approximation can and cannot achieve in PO settings.
\end{abstract}

\begingroup
\renewcommand{\thefootnote}{}
\footnotetext{
Codes and supplementary materials are available at \url{https://github.com/mintaeshkim/fmba}.
}
\endgroup


\section{Introduction}

In PO stochastic optimal control (SOC), the controller does not directly observe the system state.
Instead, decisions must be based on past observations and control inputs, i.e., the IO history.
It is well-known that optimal control can be expressed in terms of the belief state, the posterior distribution of the current state given the IO history, which induces an exact fully observed belief-Markov decision process (belief-MDP) formulation of a POMDP \cite{smallwood1973optimal, kaelbling1998planning, saldi2017asymptotic}.
While exact, the belief is in general infinite-dimensional even for simple continuous-state systems, making it impractical to compute and store \cite{saldi2018finite}.
As a result, practical controllers for PO systems rely on finite memory.
A common architecture uses a sliding window of recent observations and inputs and selects a control action via a finite memory policy \cite{saldi2017asymptotic, kara2022near}.
Throughout this work, \emph{finite memory} refers to operating on a truncated IO history, yielding a finite-dimensional information state rather than the full IO history.
Such finite memory architectures are widely used in both classical and learning-based control, yet their theoretical justification and performance analysis remain incomplete \cite{white1994finite, subramanian2019approximate, kim2025roverfly, cai2025learning}.
A central question is when finite memory can act as a meaningful substitute for the belief state and how the resulting information loss affects closed-loop performance.

Finite memory policies in PO systems have been studied extensively \cite{saldi2017asymptotic, saldi2018finite, kara2022near, kara2024partially}.
In particular, \cite{kara2022near} established near-optimality results for finite memory policies in POMDPs under non-uniform, typical-trajectory approximation criteria, showing that small performance loss can be achieved without uniform approximation over all observation sequences.
In robotics and learning-based control, it is common to feed a finite IO window (often together with the current observation) into a learned controller \cite{cai2025learning, kim2025roverfly}.
However, from an SOC perspective, several gaps remain.
Finite memory is typically treated as a restriction on the policy class rather than as an approximation of the underlying belief process \cite{kara2022near}.
Moreover, existing results rarely provide a metric-aware, quantitative relationship between information loss and value degradation along closed-loop trajectories, and fundamental limitations are often implicit.

In this paper, we develop a metric-based theory of finite memory approximation of belief states for partially observable stochastic optimal control (POSOC).
Rather than viewing finite memory as only a policy restriction, we interpret truncated IO histories as inducing an explicit \emph{finite memory belief approximation} and measure its discrepancy from the true belief in the Wasserstein-2 metric along trajectories generated by a fixed policy.
This policy-conditional perspective avoids uniform worst-case requirements over unlikely histories and yields finite, interpretable bounds \cite{kara2022near}.
Under suitable regularity conditions, we show that a Wasserstein belief mismatch controls the performance gap between the true-belief cost functional and its finite memory counterpart when both are evaluated under the same closed-loop execution, and we lift this fixed-policy comparison to an optimal value gap bound.
We also characterize fundamental limitations of finite memory control, including the necessity of retaining input history for belief reconstruction.
Finally, we specialize the framework to LQG systems, where belief mismatch and performance mismatch can be computed in closed form, and we empirically verify the paper's central mechanism: truncating IO history induces a measurable belief mismatch that quantitatively explains performance degradation.


\section{Problem Setup and Preliminaries}
\label{sec:setup}

This section introduces the POSOC problem studied in this paper and establishes all objects and notations.

We consider an infinite-horizon discounted POMDP specified by the tuple $(\mathcal{X},\mathcal{U},\mathcal{Y},P,O,c,\gamma)$, where $\mathcal{X}\subset\mathbb{R}^n$ is the state space, $\mathcal{U}\subset\mathbb{R}^m$ is the control space, $\mathcal{Y}\subset\mathbb{R}^p$ is the observation space, $P(x' \mid x,u)$ is the transition kernel, $O(y \mid x)$ is the observation kernel, $c:\mathcal{X}\times\mathcal{U}\to\mathbb{R}$ is the stage cost, and $\gamma\in(0,1)$ is the discount factor.\footnote{We adopt a kernel-based formulation to retain generality. SDE-based models will be discussed as LQG special cases via discretization and are treated explicitly in Section~\ref{sec:lqg}.}
The system evolves according to $x_{t+1}\sim P(\cdot\mid x_t,u_t)$ and $y_t\sim O(\cdot\mid x_t)$, and the controller observes only the IO history $\zeta_t=(y_{0:t},u_{0:t-1})$ based on which it selects $u_t=\pi_t(\zeta_t)$.

The belief state associated with $\zeta_t$ is defined by
\begin{equation}
b_t := \mathbb{P}(x_t\mid \zeta_t),
\end{equation}
which is a probability measure on $\mathcal{X}$.
Note that the belief is defined as a probabilistic measure, not a state by itself.
The belief evolves via the \emph{Bayesian filter} $b_{t+1} = \Phi(b_t,u_t,y_{t+1})$, where $\Phi$ denotes the belief update operator induced by $(P,O)$. 
Using the belief state, a POMDP admits an exact reduction to a fully observable MDP on the belief space $\mathcal{P}(\mathcal{X})$, commonly referred to as the \emph{belief-MDP}.
This reduction relies on the fact that the belief $b_t$ is a sufficient statistic for control in the following sense:
any two IO histories inducing the same belief lead to identical conditional distributions over future states, observations, and costs under any control sequence.
In particular, given a belief $b_t$ and input $u_t$, the predictive distribution of the next state is uniquely determined by
$\mathbb{P}(x_{t+1}\in \cdot \mid b_t,u_t) = \int P(\cdot \mid x,u_t)\, b_t(dx)$,
and the next belief $b_{t+1}$ is obtained by applying the Bayesian filter to $(b_t,u_t,y_{t+1})$.
Consequently, the belief process $\{b_t\}$ is a controlled Markov process satisfying
\begin{equation}
\label{eq:belief-mdp}
\mathbb{P}(b_{t+1}\mid y_{0:t},u_{0:t}) = \mathbb{P}(b_{t+1}\mid b_t,u_t).
\end{equation}
Moreover, the stage cost admits the exact belief-level representation $\bar c(b_t,u_t):=\mathbb{E}_{x\sim b_t}[c(x,u_t)] = \int_{\mathcal{X}
}c(x,u_t)b_t(dx)$,
and policies defined on the belief space are equivalent to history-dependent policies.
Thus, the belief-MDP formulation incurs no approximation or loss of optimality.

As a result, once the belief is taken as the system state, the POSOC problem reduces to a fully observable discounted MDP on $\mathcal{P}(\mathcal{X})$, and standard dynamic programming arguments apply.
In particular, the optimal value function is
\begin{equation}
V^\star(b_0)
:=
\inf_{\pi}
\mathbb{E}\!\left[
\sum_{t=0}^{\infty}\gamma^t \bar c(b_t,u_t)
\;\middle|\; b_0
\right],
\end{equation}
where the infimum is over belief-based policies $u_t=\pi(b_t)$.
The associated Bellman operator is
\begin{equation}
(\mathcal{T}V)(b)
=
\inf_{u\in\mathcal{U}}
\left\{
\bar c(b,u)
+
\gamma\,\mathbb{E}\!\left[V(\Phi(b,u,y'))\right]
\right\},
\end{equation}
where $y'$ is distributed according to the predictive observation law induced by $(b,u)$.
Policies defined on the full IO history and policies defined on the belief state are equivalent representations of the same decision rule, and throughout the paper we adopt the belief-based representation for simplicity.

In the following sections, to quantify belief approximation errors, we work on the Wasserstein metric space
\begin{equation}
\label{eq:Wasserstein-metric-space}
\mathcal{P}_2(\mathcal{X}) :=
\left\{
\mu\in\mathcal{P}(\mathcal{X})
\;\middle|\;
\mathbb{E}_{x\sim\mu}[\|x\|^2]\in (0, \infty)
\right\},
\end{equation}
equipped with the Wasserstein-2 distance $W_2$.

For a memory length $H$, we define the truncated IO history
\begin{equation}
\zeta_t^{(H)} := (y_{t-H:t},u_{t-H:t-1}),
\end{equation}
and the corresponding \emph{finite memory belief approximation}
\begin{equation}
\hat b_t^{(H)} := \mathbb{P}(x_t\mid \zeta_t^{(H)}).
\end{equation}

These objects fully specify the belief approximation framework used in the remainder of the paper.


\section{Finite Memory Belief Approximation}
\label{sec:finite-memory}

This section analyzes the finite memory belief approximation $\hat b_t^{(H)}$ defined in Section~\ref{sec:setup} and establishes an exponential bound on the policy-conditional belief approximation error induced by using truncated IO history.
Throughout this section, all beliefs take values in $\mathcal{P}_2(\mathcal{X})$ equipped with the Wasserstein-2 distance $W_2$.

For a fixed policy $\pi$, let $b_t^\pi$ denote the true belief process induced by the closed-loop trajectory under $\pi$, and let $\hat b_t^{(H),\pi}$ denote the finite memory belief approximation induced by the same truncated IO history.
Both beliefs are defined on the same probability space and differ only by $\sigma$-algebras.

We define the policy-conditional finite memory belief approximation error as
\begin{equation}
\label{eq:belief-approx-error}
\varepsilon_H(\pi)
:=
\sup_{t\ge 0}
\mathbb{E}_\pi\!\left[
W_2\!\left(b_t^\pi,\hat b_t^{(H),\pi}\right)
\right].
\end{equation}
This definition evaluates approximation quality only along trajectories realized under the closed-loop distribution induced by $\pi$ and avoids uniform supremum over IO histories.
\emph{Uniform approximation} is a natural but overly restrictive approach, which controls the worst-case discrepancy between the belief state and its finite memory approximation over all possible histories.
In stochastic systems, the space of feasible IO histories is vast, and many such trajectories occur with vanishing probability under closed-loop feedback control.
Uniform approximation therefore impose unnecessarily strong requirements, effectively demanding accurate approximation even along exponentially rare sample paths.
In particular, uniform approximation implicitly requires pathwise stability of belief update under control, which is generally unavailable in controlled settings \cite{kara2022near, kara2024partially}.

To obtain finite and policy-independent constants in the bounds below, and to exclude degenerate behaviors unrelated to information truncation, we restrict attention to a stabilizing policy class $\Pi_{\mathrm{stab}}$ under which all belief processes have uniformly bounded second moments.

The effect of finite memory truncation is governed by how rapidly the belief update forgets remote information under closed-loop operation.

\begin{assumption}[Controlled forgetting on $\Pi_{\mathrm{stab}}$]
\label{assump:controlled-forgetting}
There exist constants $\rho\in(0,1)$ and $C_\pi\in(0,\infty)$ such that for any $\pi\in\Pi_{\mathrm{stab}}$, any initial beliefs $\mu,\nu\in\mathcal{P}_2(\mathcal{X})$, and all $t\ge 0$,
\begin{equation}
\mathbb{E}\!\left[ W_2\!\left(
\Phi_t^\pi(\mu),\,
\Phi_t^\pi(\nu)
\right)
\;\middle|\;
u_{0:t-1}
\right]
\le
C_\pi\,\rho^t\,W_2(\mu,\nu),
\end{equation}
where $\Phi_t^\pi(\cdot)$ denotes the $t$-step belief update driven by the realized IO sequence generated by $\pi$.
\end{assumption}
\noindent
Assumption~\ref{assump:controlled-forgetting} expresses exponential stability of the belief update conditional on the realized input sequence and does not require pointwise contraction of the belief update.

For $t\ge H$, define the boundary belief at time $t-H$ by
\begin{equation}
\tilde b_{t-H}^\pi := \mathbb{P}_\pi(x_{t-H}\mid y_{t-H}).    
\end{equation}
Then $\hat b_t^{(H),\pi}$ is obtained by initializing the belief recursion at time $t-H$ with $\tilde b_{t-H}^\pi$ and applying the belief update operator along the realized sequence $(u_{t-H:t-1},y_{t-H+1:t})$.

\begin{lemma}[Finite memory belief representation]
\label{lem:window-representation}
For all $t\ge H$, the finite memory belief approximation satisfies
$\hat b_t^{(H),\pi} = \Phi_H^\pi(\tilde b_{t-H}^\pi)$,
where $\Phi_H^\pi(\cdot)$ denotes the $H$-step belief update driven by the realized IO sequence.
\end{lemma}
\begin{proof}
By the Markov property of the controlled state process and Bayes' rule, conditioning on $(y_{t-H:t},u_{t-H:t-1})$ is equivalent to conditioning on $y_{t-H}$ to initialize the belief at time $t-H$ and then applying the Bayesian filter recursively along the suffix $(u_{t-H:t-1},y_{t-H+1:t})$.
\end{proof}

\begin{lemma}[Moment bound implies Wasserstein bound]
\label{lem:w2-moment}
If $\mu,\nu\in\mathcal{P}_2(\mathcal{X})$, by (\ref{eq:Wasserstein-metric-space}), there exist $M\in(0,\infty)$ such that $\mathbb{E}_{x\sim\mu}[\|x\|^2] \le M$ and $\mathbb{E}_{x\sim\nu}[\|x\|^2] \le M$, then $W_2(\mu,\nu)\le 2\sqrt{M}$.
\end{lemma}
\begin{proof}
Let $X\sim\mu$ and $Y\sim\nu$ be independent.
Then, by properties of $W_2$ metric, $W_2(\mu,\nu)^2 \le \mathbb{E}[\|X-Y\|^2] \le 2\mathbb{E}[\|X\|]^2+2\mathbb{E}[\|Y\|^2] \le 4M$.
\end{proof}

\begin{lemma}[Forgetting implies finite memory accuracy]
\label{lem:forgetting-implies-eps}
Suppose Assumption~\ref{assump:controlled-forgetting} holds for some $\pi\in\Pi_{\mathrm{stab}}$.
Then there exists a constant $C'_\pi\in(0,\infty)$ such that for all $H\ge 0$,
\begin{equation}
\varepsilon_H(\pi)\le C'_\pi\,\rho^H.
\end{equation}
\end{lemma}

\begin{proof}
Fix $H\ge0$ and $t\ge H$.
By Lemma~\ref{lem:window-representation}, both $b_t^\pi$ and $\hat b_t^{(H),\pi}$ are obtained by applying the same $H$-step belief update to initial beliefs $b_{t-H}^\pi$ and $\tilde b_{t-H}^\pi$, respectively, along the same realized IO sequence.
Applying Assumption~\ref{assump:controlled-forgetting} conditional on this window yields
\begin{equation}
\mathbb{E}_\pi\!\left[
W_2\!\left(b_t^\pi,\hat b_t^{(H),\pi}\right)
\right]
\le
C_\pi\,\rho^H\,
\mathbb{E}_\pi\!\left[
W_2\!\left(b_{t-H}^\pi,\tilde b_{t-H}^\pi\right)
\right].
\end{equation}
Since $\pi\in\Pi_{\mathrm{stab}}$, both beliefs $b_t^\pi$ and $\hat b_t^{(H),\pi}$ have uniformly bounded second moments, so Lemma~\ref{lem:w2-moment} implies $\sup_t\mathbb{E}_\pi[W_2(b_{t-H}^\pi,\tilde b_{t-H}^\pi)]\in (0, \infty)$.
Absorbing this bound into the constant yields the claimed inequality for all $t\ge H$.
For $t<H$, the same moment bound applies and $\rho^H\le1$ allows absorption into the same constant.
Taking the supremum over $t\ge0$ completes the proof.
\end{proof}

Before proceeding to performance guarantees, we clarify that finite memory must retain both observation and input histories. Otherwise, belief reconstruction fails even in simple controlled systems.

\begin{proposition}[Necessity of input history]
\label{prop:input-necessity}
There exist POSOC systems for which no controller depending only on a finite observation window $(y_{t-H:t})$ can uniquely determine the posterior $\mathbb{P}(x_t\mid y_{t-H:t})$ independently of past inputs $(u_{t-H:t-1})$.
\end{proposition}

Section~\ref{sec:guarantees} uses only the quantity $\varepsilon_H(\pi)$ and the exponential decay established above to convert information loss into a performance loss bound, without invoking any global regularity of the Bellman operator or the value function.


\section{Performance Guarantees of a Finite Memory Belief Approximation-Based Policy}
\label{sec:guarantees}

In this section, we bound the performance loss induced by finite memory belief approximation via a policy-conditional belief mismatch evaluated under a fixed policy.
We compare two cost functionals evaluated under the same closed-loop execution induced by a single policy and differing only in the belief argument inside the belief-level stage cost.
\footnote{Policy induced by belief approximation is always suboptimal comparing to one induced by true belief. Cost comparison under same policy and inputs provides intermediate step for actual comparison between $J^\star$ and $\hat{J}_H(\pi_H^\star)$.}
This avoids comparing two different closed-loop trajectories and does not invoke any global regularity of the Bellman operator or the value function.

Fix a memory length $H\ge 0$ and a belief-based policy $\pi$.
The processes $\{b_t^\pi\}$ and $\{\hat b_t^{(H),\pi}\}$ are defined on the same probability space and are coupled through the same realized IO sequence.
Throughout this section, as mentioned in the footnote, we evaluate both cost functionals along the same realized input sequence $u_t=\pi(b_t^\pi)$, and we only change the belief parameter inside the cost $\bar c(\cdot,u_t)$.

We define the true-belief cost functional under $\pi$ by
\begin{equation}
\label{eq:true-belief-cost}
J(\pi) := \mathbb{E}_\pi\!\left[ \sum_{t=0}^\infty \gamma^t \bar c\!\left(b_t^\pi,u_t\right) \right],
\end{equation}
and define the finite memory belief approximation cost functional under the same $\pi$ and inputs $u_t=\pi(b_t^\pi)$ by
\begin{equation}
\label{eq:approx-belief-cost}
\hat J_H(\pi) := \mathbb{E}_\pi\!\left[ \sum_{t=0}^\infty \gamma^t \bar c\!\left(\hat b_t^{(H),\pi},u_t\right) \right].
\end{equation}
In particular, $J(\pi)$ and $\hat J_H(\pi)$ are evaluated under the same distribution over $(x_{0:\infty},y_{0:\infty},u_{0:\infty})$ induced by $\pi$, and they differ only by replacing $b_t^\pi$ with $\hat b_t^{(H),\pi}$ inside $\bar c(\cdot,u_t)$.

Recall the finite memory belief mismatch under $\pi$ by the policy-conditional error, defined in (\ref{eq:belief-approx-error}).
The remainder of this section shows how the belief mismatch $\varepsilon_H(\pi)$ translates into a performance gap. Section~\ref{sec:finite-memory} is used only to upper bound $\varepsilon_H(\pi)$ as a function of $H$.

To convert belief mismatch into a quantitative performance bound under a fixed closed-loop execution, we impose mild policy-conditional regularity conditions ensuring finiteness of moments and local smoothness of the belief-level cost.
\begin{assumption}[Policy-conditional state regularity]
\label{assump:state-moment-regularity}
For a fixed admissible policy $\pi$, there exists a constant $M_\pi\in (0, \infty)$ for both measure $b_t^\pi$ and $\hat{b}_t^{(H),\pi}$ such that
\begin{equation}
\sup_{t\ge 0}\,
\mathbb{E}_{x\sim b_t^\pi}\!\left[\|x\|^2\right] \le M_\pi, \quad
\sup_{t\ge 0}\, \mathbb{E}_{x\sim \hat b_t^{(H),\pi}}\!\left[\|x\|^2\right] \le M_\pi.
\end{equation}
\end{assumption}

\begin{assumption}[Quadratic growth and smoothness of cost]
\label{assump:cost-regularity}
There exists a constant $K_c>0$ and $K_g>0$ such that for all $(x,u)\in\mathcal{X}\times\mathcal{U}$,
\begin{align}
|c(x,u)| \le K_c\left(1+\|x\|^2+\|u\|^2\right), \\
\|\nabla_x c(x,u)\| \le K_g\left(1+\|x\|+\|u\|\right).
\end{align}
\end{assumption}

Assumption~\ref{assump:state-moment-regularity} ensures finiteness of the constants below under the fixed closed-loop induced by $\pi$.
Assumption~\ref{assump:cost-regularity} is compatible with quadratic costs and it does not require $c$ to be globally Lipschitz in $x$.

\begin{lemma}[Belief-level cost sensitivity under $W_2$]
\label{lem:cost-sensitivity}
Suppose Assumptions~\ref{assump:state-moment-regularity} and~\ref{assump:cost-regularity} hold for a fixed policy $\pi$.
Then there exists a finite constant $L_\pi\in (0, \infty)$ such that for any $u\in\mathcal{U}$ and any $b,\tilde b\in\mathcal{P}_2(\mathcal{X})$ satisfying
$\mathbb{E}_{x\sim b}[\|x\|^2] \le M_\pi$ and
$\mathbb{E}_{x\sim \tilde b}[\|x\|^2] \le M_\pi$,
\begin{equation}
\left|\bar c(b,u)-\bar c(\tilde b,u)\right|
\le
L_\pi\left(1+\|u\|^2\right)W_2(b,\tilde b).
\end{equation}
\end{lemma}

\begin{proof}
Fix $u\in\mathcal{U}$ and $b,\tilde b\in\mathcal{P}_2(\mathcal{X})$ satisfying the stated second-moment bounds.
Let $(X,\tilde X)$ be any coupling of $(b,\tilde b)$.
Define $X_\lambda:=\tilde X+\lambda(X-\tilde X)$ for $\lambda\in[0,1]$.
By the fundamental theorem of calculus,
\begin{equation}
c(X,u)-c(\tilde X,u)
=
\int_0^1 \nabla_x c(X_\lambda,u)^\top (X-\tilde X)\,d\lambda.
\end{equation}
Taking absolute values and applying Cauchy-Schwarz yields
\begin{equation}
|c(X,u)-c(\tilde X,u)|
\le
\left(\int_0^1 \|\nabla_x c(X_\lambda,u)\|\,d\lambda\right)\|X-\tilde X\|.
\end{equation}
Taking expectation and applying Cauchy-Schwarz gives
\begin{multline}
\mathbb{E}\!\left[|c(X,u)-c(\tilde X,u)|\right]
\le \\
\left(
\mathbb{E}\!\left[\left(\int_0^1 \|\nabla_x c(X_\lambda,u)\|\,d\lambda\right)^2\right]
\right)^{\!1/2}
\left(\mathbb{E}\|X-\tilde X\|^2\right)^{\!1/2}.
\end{multline}
By Assumption~\ref{assump:cost-regularity},
\begin{equation}
\|\nabla_x c(X_\lambda,u)\|
\le
K_g\left(1+\|X_\lambda\|+\|u\|\right).
\end{equation}
Moreover, $\|X_\lambda\|\le \|\tilde X\|+\|X-\tilde X\|$ implies
\begin{equation}
1+\|X_\lambda\|+\|u\|
\le
1+\|\tilde X\|+\|X-\tilde X\|+\|u\|.
\end{equation}
By Jensen's inequality,
\begin{equation}
\left(\int_0^1 \|\nabla_x c(X_\lambda,u)\|\,d\lambda\right)^2
\le
\int_0^1 \|\nabla_x c(X_\lambda,u)\|^2\,d\lambda,
\end{equation}
and hence
\begin{multline}
\mathbb{E}\!\left[\left(\int_0^1 \|\nabla_x c(X_\lambda,u)\|\,d\lambda\right)^2\right]
\le \\
K_g^2\,
\mathbb{E}\!\left[
\int_0^1 \left(1+\|X_\lambda\|+\|u\|\right)^2 d\lambda
\right].
\end{multline}
Using $(a+b+c)^2\le 3(a^2+b^2+c^2)$ and the bound on $\|X_\lambda\|$,
\begin{multline}
\left(1+\|X_\lambda\|+\|u\|\right)^2
\le
3\left(1+\|u\|^2+\|X_\lambda\|^2\right)
\le \\
3\left(1+\|u\|^2+2\|\tilde X\|^2+2\|X-\tilde X\|^2\right).
\end{multline}
Therefore,
\begin{multline}
\label{eq:sec-4-eq-27}
\mathbb{E}\!\left[\left(\int_0^1 \|\nabla_x c(X_\lambda,u)\|\,d\lambda\right)^2\right]
\le
3K_g^2\left(1+\|u\|^2\right)
+ \\
6K_g^2\,\mathbb{E}[\|\tilde X\|^2]
+
6K_g^2\,\mathbb{E}[\|X-\tilde X\|^2].
\end{multline}
Since $X\sim b$ and $\tilde X\sim \tilde b$ satisfy $\mathbb{E}[\|X\|^2] \le M_\pi$ and $\mathbb{E}[\|\tilde X\|^2] \le M_\pi$, we also have
\begin{equation}
\label{eq:sec-4-eq-28}
\mathbb{E}[\|X-\tilde X\|^2] \le 2\mathbb{E}[\|X\|^2] + 2\mathbb{E}[\|\tilde X\|]^2 \le 4M_\pi.
\end{equation}
Combining (\ref{eq:sec-4-eq-27}) and (\ref{eq:sec-4-eq-28}) yields the uniform bound
\begin{equation}
\mathbb{E}\!\left[\left(\int_0^1 \|\nabla_x c(X_\lambda,u)\|\,d\lambda\right)^2\right]
\le
\tilde L_\pi^2\left(1+\|u\|^2\right),
\end{equation}
where one may take $\tilde L_\pi:=K_g\sqrt{3+24M_\pi}$.
Substituting into the Cauchy-Schwarz bound gives
\begin{multline}
\mathbb{E}\!\left[|c(X,u)-c(\tilde X,u)|\right]
\le \\
\tilde L_\pi\left(1+\|u\|^2\right)^{1/2}\left(\mathbb{E}[\|X-\tilde X\|^2]\right)^{1/2}.
\end{multline}
Since $\left|\bar c(b,u)-\bar c(\tilde b,u)\right|\le \mathbb{E}\!\left[|c(X,u)-c(\tilde X,u)|\right]$ and $(1+\|u\|^2)^{1/2}\le 1+\|u\|^2$, we obtain
\begin{equation}
\left|\bar c(b,u)-\bar c(\tilde b,u)\right|
\le
\tilde L_\pi\left(1+\|u\|^2\right)\left(\mathbb{E}\|X-\tilde X\|^2\right)^{1/2}.
\end{equation}
Taking the infimum over all couplings $(X,\tilde X)$ yields
\begin{equation}
\left|\bar c(b,u)-\bar c(\tilde b,u)\right|
\le
\tilde L_\pi\left(1+\|u\|^2\right)W_2(b,\tilde b).
\end{equation}
Setting $L_\pi:=\tilde L_\pi$ completes the proof.
\end{proof}

\begin{lemma}[Fixed-policy performance mismatch]
\label{lem:fixed-policy-mismatch}
Suppose again Assumptions~\ref{assump:state-moment-regularity} and~\ref{assump:cost-regularity} hold for a fixed policy $\pi$ and suppose additionally that $\sup_{t\ge 0}\mathbb{E}_\pi[\|u_t\|^2]\in (0, \infty)$ for the closed-loop inputs $u_t=\pi(b_t^\pi)$.
Assume further that there exists a constant $U_\pi\in (0, \infty)$ such that $\|u_t\|^2 \le U_\pi \quad\text{a.s. for all } t\ge 0$. Bounded input assumption is reasonable in most optimal control problems.
Then there exists a finite constant $C_\pi\in (0, \infty)$ such that for all $H\ge 0$,
\begin{equation}
\left|J(\pi)-\hat J_H(\pi)\right| \le \frac{C_\pi}{1-\gamma}\,\varepsilon_H(\pi).
\end{equation}
\end{lemma}

\begin{proof}
By the definitions of $J(\pi)$ and $\hat J_H(\pi)$ (See (\ref{eq:true-belief-cost}) and (\ref{eq:approx-belief-cost})) and the triangle inequality,
\begin{equation}
\left|J(\pi)-\hat J_H(\pi)\right|
\le
\sum_{t=0}^{\infty}
\gamma^t\,
\mathbb{E}_\pi\!\left[
\left|\bar c(b_t^\pi,u_t)-\bar c(\hat b_t^{(H),\pi},u_t)\right|
\right].
\end{equation}
By Lemma~\ref{lem:cost-sensitivity},
\begin{equation}
\left|\bar c(b_t^\pi,u_t)-\bar c(\hat b_t^{(H),\pi},u_t)\right|
\le
L_\pi\left(1+\|u_t\|^2\right)
W_2\!\left(b_t^\pi,\hat b_t^{(H),\pi}\right).
\end{equation}
Taking expectations yields
\begin{multline}
\mathbb{E}_\pi\!\left[
\left|\bar c(b_t^\pi,u_t)-\bar c(\hat b_t^{(H),\pi},u_t)\right|
\right]
\le \\
L_\pi\,
\mathbb{E}_\pi\!\left[
\left(1+\|u_t\|^2\right)
W_2\!\left(b_t^\pi,\hat b_t^{(H),\pi}\right)
\right].
\end{multline}
By uniform input boundedness, we have
\begin{multline}
\mathbb{E}_\pi\!\left[
\left(1+\|u_t\|^2\right)
W_2\!\left(b_t^\pi,\hat b_t^{(H),\pi}\right)
\right]
\le \\
\left(1+U_\pi\right)
\mathbb{E}_\pi\!\left[
W_2\!\left(b_t^\pi,\hat b_t^{(H),\pi}\right)
\right]
\le
\left(1+U_\pi\right)\varepsilon_H(\pi).
\end{multline}
Combining the above inequalities and summing over $t\ge 0$ yields
\begin{equation}
\left|J(\pi)-\hat J_H(\pi)\right|
\le
L_\pi\left(1+U_\pi\right)
\varepsilon_H(\pi)
\sum_{t=0}^{\infty}\gamma^t
=
\frac{C_\pi}{1-\gamma}\,\varepsilon_H(\pi),
\end{equation}
where $C_\pi:=L_\pi\left(1+U_\pi\right)$.
\end{proof}

We now lift the fixed-policy mismatch bound to the optimal value gap bound between the belief-optimal controller and the optimal finite memory controller.
Let $\pi^\star$ denote an optimal policy for the belief-MDP and let $\pi_H^\star$ denote an optimal policy among finite memory belief approximation-based policies measurable with respect to $\zeta_t^{(H)}$.

\begin{theorem}[Performance bound via belief mismatch]
\label{thm:value-gap}
Suppose Assumptions~\ref{assump:state-moment-regularity} and~\ref{assump:cost-regularity} hold for $\pi^\star$ and suppose additionally that input is uniformly bounded,
\begin{equation}
\sup_{t\ge 0}\mathbb{E}_{\pi^\star}[\|u_t\|^2]\in (0, \infty)
\end{equation}
for the closed-loop inputs $u_t=\pi^\star(b_t^{\pi^\star})$.
Then for every $H\ge 0$,
\begin{equation}
0\le J^\star-J_H^\star
\le
\frac{C_{\pi^\star}}{1-\gamma}\,\varepsilon_H(\pi^\star),
\end{equation}
where $J^\star:=J(\pi^\star)$ and $J_H^\star:=\hat J_H(\pi_H^\star)$.
\end{theorem}

\begin{proof}
Since $\pi_H^\star$ minimizes $\hat J_H(\cdot)$ over the finite memory policy class, we have
\begin{equation}
\hat J_H(\pi_H^\star)\le \hat J_H(\pi^\star).
\end{equation}
Therefore,
\begin{equation}
J^\star-J_H^\star
=
J(\pi^\star)-\hat J_H(\pi_H^\star)
\le
J(\pi^\star)-\hat J_H(\pi^\star).
\end{equation}
Applying Lemma~\ref{lem:fixed-policy-mismatch} to $\pi=\pi^\star$ yields
\begin{equation}
J(\pi^\star)-\hat J_H(\pi^\star)
\le
\frac{C_{\pi^\star}}{1-\gamma}\,\varepsilon_H(\pi^\star),
\end{equation}
which proves the claim.
\end{proof}

\begin{corollary}[Exponential decay and forgetting]
\label{cor:exp-decay}
Suppose the conditions of Theorem~\ref{thm:value-gap} hold and suppose in addition that there exist constants $C'_{\pi^\star}\in (0, \infty)$ and $\rho\in(0,1)$ such that
\begin{equation}
\varepsilon_H(\pi^\star)\le C'_{\pi^\star}\rho^H
\end{equation}
for all $H\ge 0$.
Then for all $H\ge 0$,
\begin{equation}
J^\star-J_H^\star
\le
\frac{C_{\pi^\star}C'_{\pi^\star}}{1-\gamma}\,\rho^H.
\end{equation}
\end{corollary}

\begin{proof}
The bound follows by substituting $\varepsilon_H(\pi^\star)\le C'_{\pi^\star}\rho^H$ into Theorem~\ref{thm:value-gap}.
\end{proof}

Corollary~\ref{cor:exp-decay} shows that the finite memory optimality gap decays exponentially in $H$ whenever $\varepsilon_H(\pi^\star)$ decays exponentially in $H$.
Section~\ref{sec:finite-memory} provides sufficient conditions for such exponential decay via controlled forgetting of belief update under stabilizing policies.


\section{LQG Specialization and Numerical Results}
\label{sec:lqg}

In this section, we empirically validates the theoretical results from Sections~\ref{sec:finite-memory} and \ref{sec:guarantees} using an LQG system.
Our goal is to test whether truncated IO history induces a belief mismatch and whether this mismatch explains performance gap as predicted by the theory.
Although belief computation in LQG admits a closed-form Kalman recursion, this tractability is used here strictly as an experimental advantage.
Kalman filter allows exact evaluation of the true belief, the finite memory belief approximation, and the Wasserstein discrepancy between them.
Finite memory belief approximation is constructed to discard past information and therefore induces a loss of information as $H$ decreases.

We consider the PO linear stochastic system
\begin{align}
x_{t+1} &= A x_t + B u_t + w_t,
\qquad w_t \sim \mathcal{N}(0,\Sigma_w), \\
y_t &= C x_t + v_t,
\qquad v_t \sim \mathcal{N}(0,\Sigma_v),
\end{align}
with Gaussian prior $x_0 \sim \mathcal{N}(m_0,P_0)$ and quadratic cost,
\begin{equation}
c(x_t,u_t) = x_t^\top Q x_t + u_t^\top R u_t,
\end{equation}
and performance is evaluated using the discounted infinite-horizon objective defined in (\ref{eq:true-belief-cost}) and (\ref{eq:approx-belief-cost}).

The true belief $b_t=\mathcal{N}(m_t,P_t)$ is obtained by the Kalman filter, and the control policy is fixed to the LQG controller
\begin{equation}
u_t = -K m_t,
\end{equation}
where $K$ is the infinite-horizon LQR gain.
Throughout this section, the closed-loop execution is always generated by this policy, and only the belief argument inside the belief-level cost is modified, exactly matching the fixed-policy comparison analyzed theoretically.

The finite memory belief approximation $\hat b_t^{(H)}$ is implemented using
a window-restart construction consistent with Lemma~\ref{lem:window-representation}.
For each $t \ge H$, the belief recursion is reinitialized at time $s=t-H$
using a boundary belief $\tilde b_s$ depends only on the single observation $y_s$.
In the implementation, $\tilde b_s$ is obtained by performing a Kalman measurement update from the fixed prior $(m_0,P_0)$ using $y_s$.
The belief is then propagated forward for $H$ steps using only the truncated IO $(u_{s:s+H-1},y_{s+1:s+H})$ to obtain $\hat b_t^{(H)}$.

Since both $b_t$ and $\hat b_t^{(H)}$ are Gaussian,the belief mismatch is computed using the closed-form Wasserstein-2 distance,
\begin{multline}
W_2^2(b_t,\hat b_t^{(H)}) =
\|m_t-\hat m_t^{(H)}\|^2
+ \\
\mathrm{Tr}\Big(
P_t+\hat P_t^{(H)}
-2\big(P_t^{1/2}\hat P_t^{(H)}P_t^{1/2}\big)^{1/2}
\Big).
\end{multline}
The policy-conditional error $\varepsilon_H(\pi)$
is estimated via Monte-Carlo averaging over multiple rollouts.

The system is instantiated as an LQG double integrator,
\begin{equation}
A =
\begin{bmatrix}
1 & \Delta t\\
0 & 1
\end{bmatrix},
\quad
B =
\begin{bmatrix}
\frac{1}{2}\Delta t^2\\
\Delta t
\end{bmatrix},
\quad
C =
\begin{bmatrix}
1 & 0
\end{bmatrix},
\end{equation}
so that only the position is observed.
All experiments use a fixed controller and sweep the memory length $H \in \{0,1,2,5,10,20,50,100\}$, with horizon $T=1000$ and 50 random seeds.

Figure~\ref{fig:lqg-eps-vs-H} plots the estimated belief mismatch $\varepsilon_H(\pi)$ as a function of the memory length $H$.
Consistent with Lemma~\ref{lem:forgetting-implies-eps},
the mismatch decays approximately exponentially in $H$,
appearing as a linear trend on log axes.

\begin{figure}[t]
    \centering
    \includegraphics[width=0.8\linewidth]{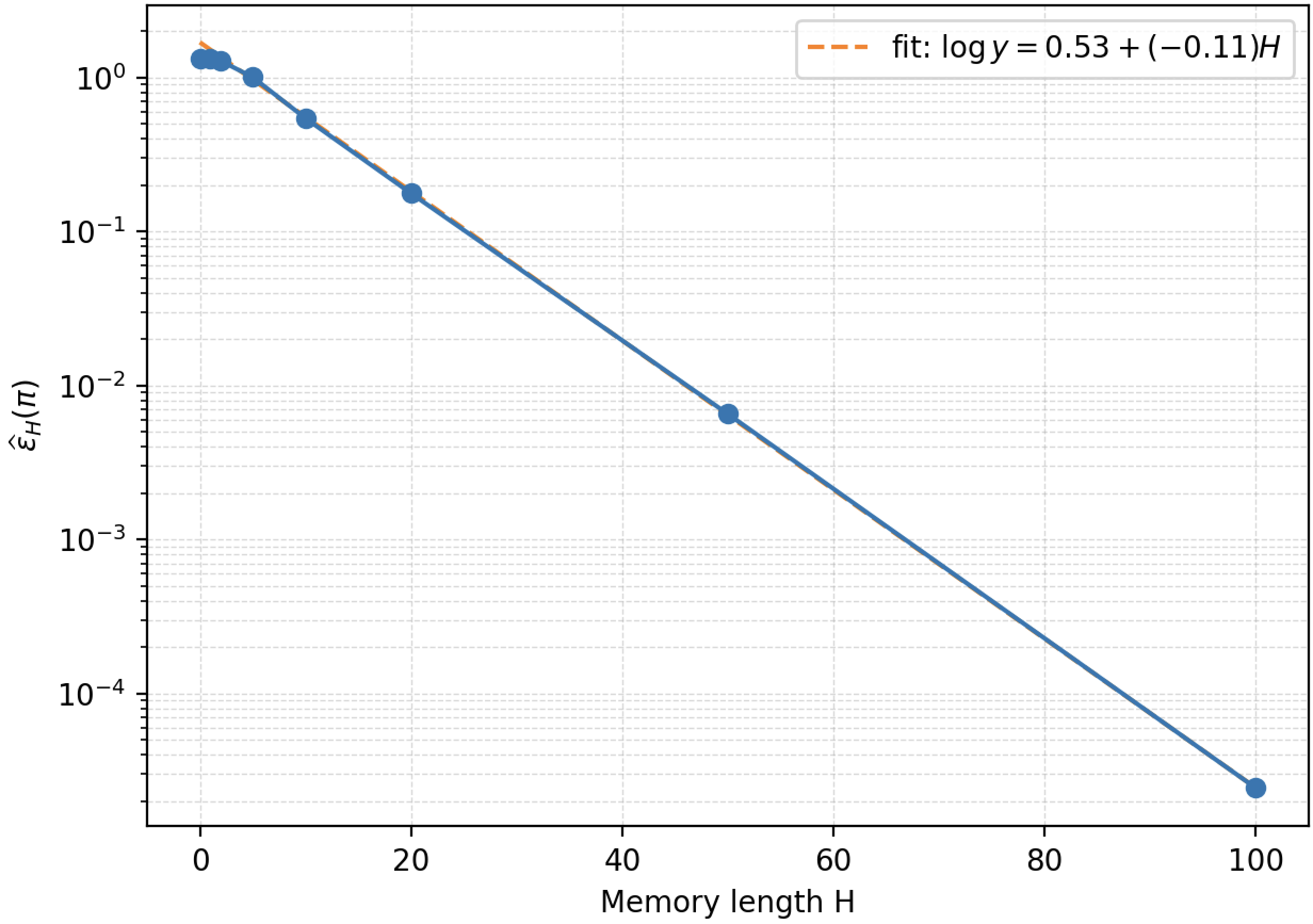}
    \caption{Belief mismatch $\varepsilon_H(\pi)$ versus memory length $H$ in log scale (y-axis).
    The approximately linear decay confirms exponential forgetting under closed-loop operation.}
    \label{fig:lqg-eps-vs-H}
\end{figure}

Figure~\ref{fig:lqg-gap-vs-eps} examines the relationship between belief mismatch and fixed-policy performance gap.
As predicted by Lemma~\ref{lem:fixed-policy-mismatch}, the cost gap scales approximately linearly with $\varepsilon_H(\pi)$,
which appears as a linear trend on log axes.

\begin{figure}[t]
    \centering
    \includegraphics[width=0.8\linewidth]{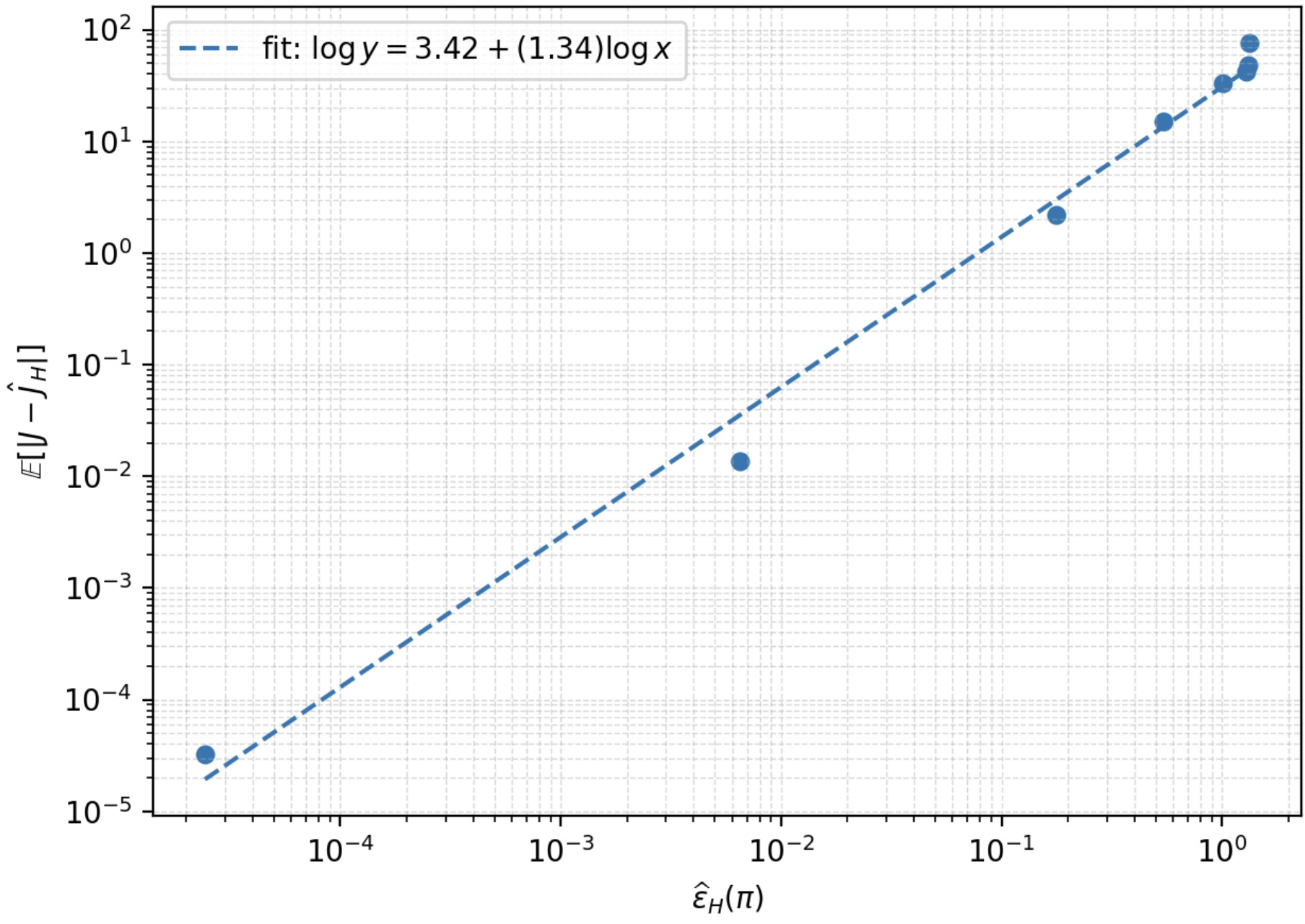}
    \caption{Cost mismatch versus belief mismatch under fixed-policy in log-log scale.
    The observed linear scaling supports the theoretical bound
    $|J(\pi)-\hat J_H(\pi)| \propto \varepsilon_H(\pi)$.}
    \label{fig:lqg-gap-vs-eps}
    \vspace{-4mm}
\end{figure}

For diagnostic purposes, Figure~\ref{fig:lqg-w2-time} visualizes the time evolution of the belief mismatch for representative values of $H$.
The mismatch is largest during early transients and stabilizes after, illustrating how finite memory primarily affects the filter’s ability to accumulate information over time.

\begin{figure}[t]
    \centering
    \includegraphics[width=0.8\linewidth]{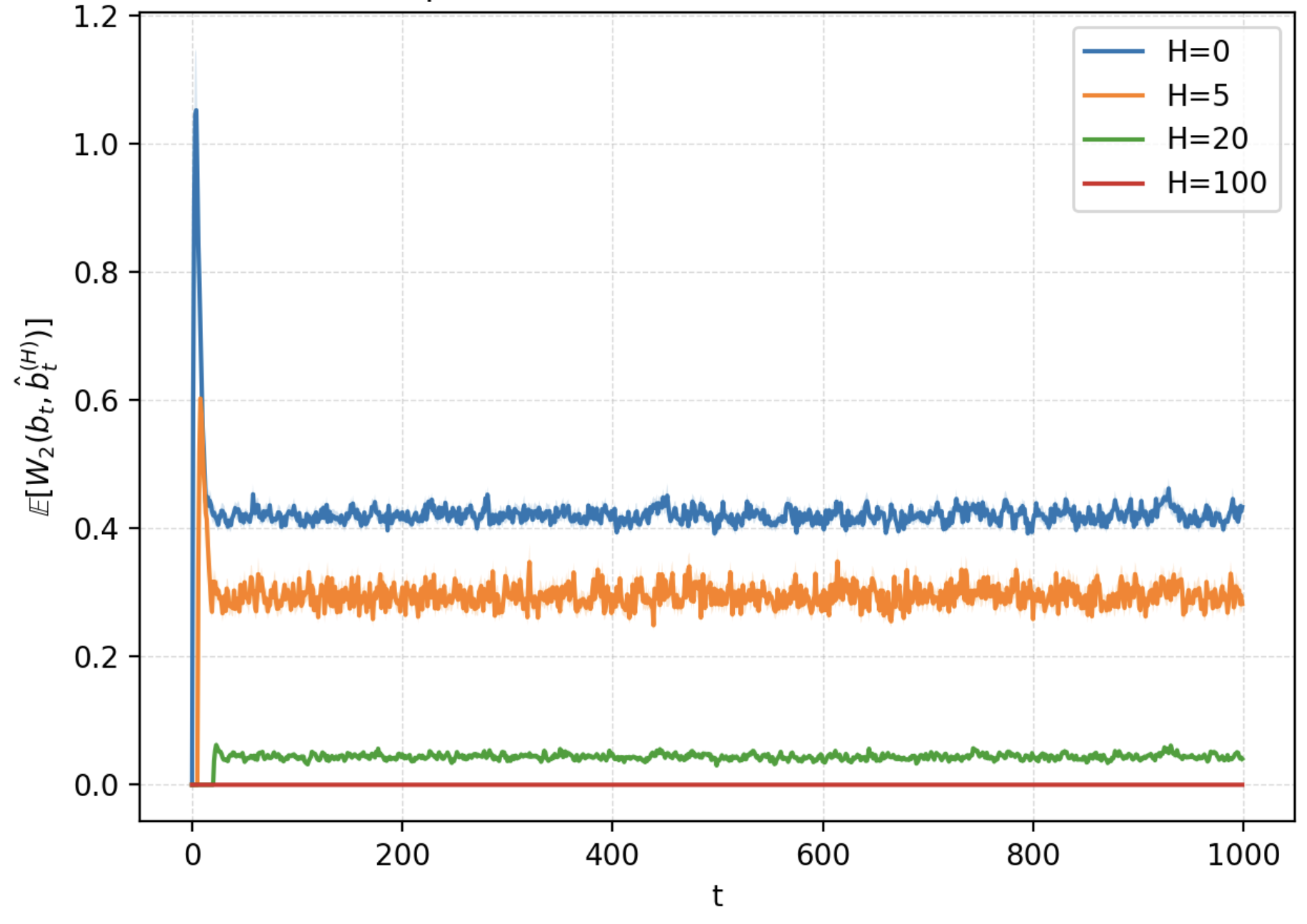}
    \caption{Time profile of the belief mismatch    $W_2(b_t,\hat b_t^{(H)})$ for selected memory lengths.
    Curves show mean $\pm$ standard error over 50 random seeds.}
    \label{fig:lqg-w2-time}
    \vspace{-4mm}
\end{figure}

Overall, the LQG experiments provide a concrete validation of the paper’s central mechanism.
Truncated IO history induces a measurable belief mismatch, which in turn explains performance degradation under closed-loop control.
The role of LQG here is not to trivialize belief approximation, but to provide a setting in which information loss, belief error, and cost degradation can be exactly in closed-form.


\section{Conclusions}

This paper studied finite memory POSOC by interpreting truncated IO histories as inducing finite memory belief approximations.
By measuring information loss in the Wasserstein metric and evaluating performance under a fixed closed-loop execution, we established a relationship between belief mismatch and value degradation.
Under controlled forgetting, the belief approximation error decays exponentially with memory length, yielding an explicit exponential bound on the performance gap.
Our analysis shows that finite memory should be viewed as an approximation of the underlying information state rather than merely a restriction on the policy class.
We also identified fundamental limitations of finite memory control, including the necessity of retaining input history and unavoidable exponential memory requirements in general PO systems.
Specialization to LQG systems showed that these effects persist even when belief computation is tractable, and numerical experiments verified that the theoretical bounds capture the observed scaling behavior.


\addtolength{\textheight}{-12cm}


\printbibliography
\end{document}